\documentclass[11pt]{article}
\usepackage{authblk}
\usepackage[margin=1in]{geometry}
\usepackage{amsmath, amssymb, amsthm, mathtools}
\usepackage{bm}
\usepackage{enumitem}
\usepackage{hyperref}
\usepackage{microtype, comment}
\usepackage{natbib}
\hypersetup{
    colorlinks=true,
    linkcolor=blue,
    citecolor=blue,
    urlcolor=blue
}

\numberwithin{equation}{section}

\newtheorem{theorem}{Theorem}[section]
\newtheorem{assumption}{Assumption}[section]
\newtheorem{lemma}{Lemma}[section]
\newtheorem{corollary}{Corollary}[section]
\newtheorem{remark}{Remark}[section]

\newcommand{\R}{\mathbb R}
\newcommand{\Prob}{\mathbb P}
\newcommand{\E}{\mathbb E}
\newcommand{\ind}{\mathbf 1}
\newcommand{\argmax}{\operatorname*{arg\,max}}
\newcommand{\argmin}{\operatorname*{arg\,min}}
\newcommand{\tr}{\operatorname{tr}}

\newcommand{\BFDR}{\operatorname{BFDR}}
\newcommand{\FDR}{\operatorname{FDR}}

\title{Bayesian Distilled Clustering for High-Dimensional Mixture Models}
\author[1]{Pulkita Aggarwal\thanks{Corresponding author: pulkitaaggarwal09@gmail.com}}
\author[2]{Abhishek Bhattacharjee\thanks{Corresponding author: abhishek@theabstractmath.com}}

\affil[1,2]{Abstract Math Institute}
\affil[1]{\texttt{pulkitaaggarwal09@gmail.com}}
\affil[2]{\texttt{abhishek@theabstractmath.com}}

\begin{document}

\maketitle

\section{Introduction}
\label{sec:introduction}

Latent subgroup analysis is a central problem in modern statistics, with applications in genomics, precision medicine, social science, and large-scale observational studies. In these settings, the observed population is often heterogeneous, and the scientific objective is not merely to predict an outcome, but to identify interpretable subpopulations with distinct covariate structure, response behavior, or dependence patterns. Mixture models provide a natural probabilistic framework for this task. They allow the data-generating distribution to be represented as a weighted combination of subgroup-specific laws, with the subgroup label treated as an unobserved random variable.

The high-dimensional regime presents a more delicate problem. When the number of measured covariates is large, only a small subset of variables may be responsible for meaningful subgroup separation. The remaining variables may contain measurement noise, nuisance variation, or correlation induced by the relevant covariates, but they need not carry independent information about subgroup membership. Standard clustering methods often treat all variables as exchangeable inputs. In high dimensions, this can be statistically inefficient and scientifically misleading: irrelevant coordinates accumulate noise, distort distances, and may obscure the low-dimensional structure that defines the latent classes.

This paper develops a Bayesian distilled clustering framework for high-dimensional mixture models. The basic premise is that clustering should be carried out in a statistically justified subspace rather than in the full ambient space. The proposed method first uses a Bayesian variable-selection model to estimate posterior inclusion probabilities for the covariates. These probabilities quantify the evidence that each covariate contributes to subgroup separation or subgroup-specific response behavior. A distilled covariate set is then obtained by controlling the posterior expected false-discovery proportion. Clustering is subsequently performed only on this selected subspace. Finally, within each inferred subgroup, conditional-independence diagnostics are used to study subgroup-specific dependence among the selected variables.

The formulation is deliberately model-based. The selected variables are not obtained by marginal screening, nor by a generic dimension-reduction criterion. Instead, the distillation step is tied to the mixture structure and to the response model. This makes the resulting subspace directly aligned with the scientific object of interest: latent subgroups that differ in distributional structure and response behavior.

\subsection{Existing Works}
\label{subsec:existing works}

The literature on variable selection for clustering problems from a frequentist perspective includes methods based on the ideas of sparse K-means by \cite{witten2010framework}, pairwise penalized model-based clustering by \cite{guo2010pairwise}, and L1-penalized EM algorithms for finite mixtures of regressions by \cite{stadler2010l1}. The extension of the latter approach to group-lasso penalization without sample splitting was provided by \cite{wang2024statistical}. None of these methods maintain an uncertainty measure over the selected set of variables and do not provide a means of controlling the false discoveries of variables within the selection step.

The Bayesian approach to variable selection in clustering was first considered by \cite{tadesse2005bayesian}. \cite{kim2006variable} extended this approach to Dirichlet process mixtures. \cite{raftery2006variable} employed a likelihood approach to model variable selection, where models including and excluding each variable were compared using Bayes factors and greedy search procedures. \cite{maugis2009variable} considered variable selection within Dirichlet process mixtures, but with a more careful modeling of the dependence between relevant and irrelevant variables. In each of these approaches, however, no asymptotic theoretical results are provided for cases in which the number of variables p is allowed to grow at a much faster rate than the number of observations n, and where a response model is included that is specific to subgroups of the clustering algorithm results.

The posterior concentration theory for high-dimensional Bayesian variable selection was developed by \cite{castillo2012needles}, and also by \cite{castillo2015bayesian}. The theory applies to the type of spike-and-slab priors with complexity penalties that have been used in this paper. \cite{scott2010bayes} provided an analysis of the multiplicity properties of these priors. The posterior false-discovery rate rule is based upon the approaches of \cite{newton2004detecting} and \cite{muller2004optimal}.

The graphical diagnostics for the subgroups is based upon the graphical lasso method of \cite{friedman2008sparse} and the debiasing approach of \cite{jankova2015confidence}, who prove that the debiased graphical lasso estimator is asymptotically normal entry-wise. Thus, statistical tests can be performed on each graphical edge within the estimated subgroups. A false-discovery rate control rule for graphical edges is that of \cite{benjamini1995controlling}.

The approach of this paper unifies these methods in a way that does not appear in the existing literature. While the variable selection and clustering steps are separate in terms of their computational approach, the two methods are based upon the same underlying probabilistic model that encompasses both the subgroup-specific covariate distribution and the response variable. Furthermore, the choice of subspace is based upon the posterior false-discovery rate rule rather than upon model selection methods. Guarantees are established for information preservation, clustering, response-parameter estimation, and graphical model estimation, and the variable-selection step is placed on a well-defined posterior.
\subsection{Our contributions}
\label{subsec:contributions}

The main contributions are as follows.

\begin{itemize}[leftmargin=2em]
    \item We introduce a Bayesian distillation framework for high-dimensional mixture models in which feature selection and clustering are separated computationally but connected through a common probabilistic model.

    \item We define posterior inclusion probabilities for covariates and use them to construct a distilled covariate set with posterior false-discovery control. This provides a principled alternative to marginal screening or unsupervised dimension reduction.

    \item We show that, under a conditional-irrelevance assumption, the excluded covariates contain no additional subgroup information once the distilled variables are known. Consequently, the Bayes classifier based on all covariates coincides with the Bayes classifier based only on the relevant covariates.

    \item We establish theoretical properties of the procedure, including preservation of subgroup information after distillation, clustering consistency, posterior contraction for subgroup-specific response parameters, and validity of subgroup-level conditional-independence tests.

    \item We formulate a subgroup-specific graphical diagnostic based on conditional independence. In the Gaussian case, this diagnostic is expressed through the precision matrix within each inferred subgroup and allows the analyst to study how selected covariates interact differently across latent classes.

    \item We report a simulation study of variable-selection accuracy, cluster recovery, prediction, and graph recovery in high-dimensional mixture settings, with comparisons against existing methods.
\end{itemize}

\subsection{Structure of the paper}
\label{subsec:paper-structure}

The remainder of the paper is organized as follows. Section~\ref{sec:setting} gives the statistical model and defines the relevant covariate set. Section~\ref{sec:bayesian-distillation} presents the Bayesian distillation model, including the prior specification, posterior inclusion probabilities, and posterior false-discovery rule. Section~\ref{sec:clustering} describes the clustering step on the distilled subspace. Section~\ref{sec:conditional-independence} introduces subgroup-specific conditional-independence diagnostics. Section~\ref{sec:theory} states the assumptions and theoretical guarantees. Section~\ref{sec:methodology} summarizes the practical implementation. Section~\ref{sec:simulation} reports a simulation study. Section~\ref{sec:common-space} places the comparison of inclusion vectors on a common sample space. Section~\ref{sec:conclusion} concludes. Proofs are collected in Appendix~\ref{app:proofs}.

\section{Statistical Setting}
\label{sec:setting}

Let
\[
D_n=\{(Y_i,X_i):1\le i\le n\}
\]
be independent observations, where \(Y_i\in\mathcal Y\) is a response and
\[
X_i=(X_{i1},\ldots,X_{ip})^\top\in\R^p
\]
is a high-dimensional covariate vector. The dimension \(p\) may be much larger than \(n\). The population is assumed to contain \(K\) latent subgroups. For each subject \(i\), let
\[
Z_i\in\{1,\ldots,K\}
\]
denote the unobserved subgroup label, with
\[
\Prob(Z_i=k)=\pi_k,\qquad \pi_k>0,\qquad \sum_{k=1}^K\pi_k=1.
\]

The central statistical assumption is that only a small subset of covariates carries information about the latent subgroup structure and the subgroup-specific response behavior. Let
\[
S_0\subset\{1,\ldots,p\}
\]
be the unknown set of relevant covariates, with
\[
|S_0|=s_0\ll p.
\]
The complement \(S_0^c\) consists of variables that may be noisy, correlated, or biologically measurable, but do not help distinguish the latent subgroups after the variables in \(S_0\) are known.

For a candidate subset \(S\), write \(X_{iS}\) for the coordinates of \(X_i\) indexed by \(S\). A rigorous model for the joint distribution of \((Y_i,X_i)\) is
\[
p(y,x)
=
\sum_{k=1}^K
\pi_k
q_k(x_S;\psi_k)
r_k(y\mid x_S;\theta_k)
h(x_{S^c}\mid x_S),
\]
where \(q_k\) is the subgroup-specific covariate density on the selected coordinates, \(r_k\) is the subgroup-specific response model, and \(h\) is the conditional distribution of the excluded variables. The key structural restriction is that \(h\) does not depend on \(k\). Thus the excluded variables may be statistically dependent on the selected variables, but they do not contain additional subgroup information once \(X_S\) is known.

A concrete Gaussian version, useful for theory and implementation, is
\[
X_{iS_0}\mid Z_i=k
\sim
N_{s_0}(\mu_{k0},\Sigma_{k0}),
\]
and
\[
Y_i\mid X_i,Z_i=k
\sim
f\{y;\eta_{ik},\phi_k\},
\qquad
\eta_{ik}=\alpha_k+X_{iS_0}^\top\beta_{k0}.
\]
Here \(f\) may be a Gaussian, Bernoulli, Poisson, or other regular exponential-family density. The covariates outside \(S_0\) satisfy
\[
p(X_{iS_0^c}\mid X_{iS_0},Z_i=k)
=
p(X_{iS_0^c}\mid X_{iS_0})
\]
for every \(k\). Therefore, irrelevant variables may still be statistically associated with relevant variables, but they do not help separate the latent groups.

This distinction is important. A variable is not called irrelevant merely because it has small marginal variance or weak marginal association with the response. It is irrelevant only if it contributes neither to subgroup separation nor to subgroup-specific response behavior after the selected variables are accounted for.

\section{Bayesian Distillation Model}
\label{sec:bayesian-distillation}

Introduce a binary inclusion vector
\[
\gamma=(\gamma_1,\ldots,\gamma_p)^\top\in\{0,1\}^p,
\]
where
\[
\gamma_j=
\begin{cases}
1, & \text{if covariate } j \text{ is included},\\
0, & \text{otherwise}.
\end{cases}
\]
Let
\[
S_\gamma=\{j:\gamma_j=1\}.
\]

For a given \(\gamma\), define the observed-data likelihood
\[
L_n(\Theta,\gamma)
=
\prod_{i=1}^n
\sum_{k=1}^K
\pi_k
q_k(X_{iS_\gamma};\psi_k)
r_k(Y_i\mid X_{iS_\gamma};\theta_k),
\]
where
\[
\Theta=\{\pi_k,\psi_k,\theta_k:1\le k\le K\}.
\]
The variables outside \(S_\gamma\) are not used in the clustering likelihood because, under the model, their conditional distribution does not depend on the latent subgroup once \(X_{iS_\gamma}\) is given.

A sparse prior is placed on \(\gamma\). One convenient choice is the complexity prior
\[
\Prob(\gamma)
\propto
c^{-|\gamma|}p^{-a|\gamma|},
\qquad a>1,\ c>0,
\]
where
\[
|\gamma|=\sum_{j=1}^p\gamma_j.
\]
This prior penalizes large models strongly enough to prevent the posterior from including many irrelevant variables when \(p\) is large.

Conditional on \(\gamma\), the mixture weights have prior
\[
(\pi_1,\ldots,\pi_K)
\sim
\operatorname{Dirichlet}(a_\pi,\ldots,a_\pi),
\]
and the subgroup-specific parameters receive proper continuous priors. For example, in the Gaussian covariate model,
\[
\mu_{k,S_\gamma}\mid \Sigma_{k,S_\gamma}
\sim
N_{|S_\gamma|}(m_0,\kappa_0^{-1}\Sigma_{k,S_\gamma}),
\]
and
\[
\Sigma_{k,S_\gamma}^{-1}
\sim
\operatorname{Wishart}(\nu_0,A_0).
\]
For the response model,
\[
\beta_{k,S_\gamma}\sim N_{|S_\gamma|}(0,\tau_\beta^2 I),
\qquad
\alpha_k\sim N(0,\tau_\alpha^2).
\]

The posterior distribution is
\[
\Pi(\Theta,\gamma\mid D_n)
=
\frac{
L_n(\Theta,\gamma)p(\Theta\mid\gamma)p(\gamma)
}{
\sum_{\gamma'}
\int
L_n(\Theta',\gamma')p(\Theta'\mid\gamma')p(\gamma')\,d\Theta'
}.
\]

For each variable \(j\), define its posterior inclusion probability
\[
w_j
=
\Pi(\gamma_j=1\mid D_n).
\]

The distilled set is chosen by controlling the posterior expected false-discovery proportion. For a threshold \(t\in[0,1]\), define
\[
\widehat S(t)=\{j:w_j\ge t\}.
\]
The posterior expected false-discovery proportion is
\[
\BFDR(t)
=
\frac{
\sum_{j:w_j\ge t}(1-w_j)
}{
|\widehat S(t)|\vee 1
}.
\]
For a preassigned level \(q\in(0,1)\), choose
\[
\widehat t
=
\min\{w_j:\BFDR(w_j)\le q\},
\]
and set
\[
\widehat S=\widehat S(\widehat t),
\]
with \(\widehat S=\emptyset\) if no \(w_j\) satisfies \(\BFDR(w_j)\le q\). Since \(\BFDR(t)\) is the average of \(1-w_j\) over the selected variables, it does not decrease as \(t\) decreases, so \(\widehat S\) is the largest set of top-ranked variables whose posterior expected false-discovery proportion is at most \(q\). This rule has a clear interpretation: among the selected variables, the posterior expected fraction of irrelevant variables is at most \(q\).

\section{Clustering on the Distilled Subspace}
\label{sec:clustering}

After the distilled set \(\widehat S\) is obtained, the clustering step is performed only on \(X_{i\widehat S}\). The primary clustering model is
\[
X_{i\widehat S}\mid Z_i=k
\sim
q_k(\cdot;\psi_k),
\qquad
\Prob(Z_i=k)=\pi_k.
\]
In the Gaussian version,
\[
X_{i\widehat S}\mid Z_i=k
\sim
N_{|\widehat S|}(\mu_{k,\widehat S},\Sigma_{k,\widehat S}).
\]

Let
\[
\widehat\Psi
=
\{\widehat\pi_k,\widehat\mu_k,\widehat\Sigma_k:1\le k\le K\}
\]
denote the fitted mixture parameters. The posterior membership weight for observation \(i\) and subgroup \(k\) is
\[
\widehat r_{ik}
=
\frac{
\widehat\pi_k
\varphi_{|\widehat S|}
(X_{i\widehat S};\widehat\mu_k,\widehat\Sigma_k)
}{
\sum_{\ell=1}^K
\widehat\pi_\ell
\varphi_{|\widehat S|}
(X_{i\widehat S};\widehat\mu_\ell,\widehat\Sigma_\ell)
}.
\]
The estimated subgroup label is
\[
\widehat Z_i
=
\argmax_{1\le k\le K}\widehat r_{ik}.
\]

The response \(Y_i\) may be incorporated in two ways. In the primary version, \(Y_i\) is used to guide variable selection, but cluster assignment is based on the covariate subspace. This respects the goal of identifying covariate-defined latent subgroups. In a supervised extension, one may assign labels by the joint posterior weight
\[
\widehat r_{ik}^{\mathrm{joint}}
=
\frac{
\widehat\pi_k
q_k(X_{i\widehat S};\widehat\psi_k)
r_k(Y_i\mid X_{i\widehat S};\widehat\theta_k)
}{
\sum_{\ell=1}^K
\widehat\pi_\ell
q_\ell(X_{i\widehat S};\widehat\psi_\ell)
r_\ell(Y_i\mid X_{i\widehat S};\widehat\theta_\ell)
}.
\]
The covariate-only version is preferred when the scientific target is subgroup discovery. The joint version is useful when the target is response-stratified prediction.

\section{Conditional-Independence Diagnostics Within Subgroups}
\label{sec:conditional-independence}

For each estimated subgroup \(k\), define
\[
\widehat{\mathcal I}_k=\{i:\widehat Z_i=k\},
\qquad
\widehat n_k=|\widehat{\mathcal I}_k|.
\]
The goal is to study how the selected variables depend on one another within each subgroup. In the Gaussian case,
\[
X_{i\widehat S}\mid Z_i=k
\sim
N(\mu_k,\Sigma_k),
\]
and the conditional-independence structure is encoded by the precision matrix
\[
\Omega_k=\Sigma_k^{-1}.
\]
For two distinct selected variables \(a,b\in\widehat S\),
\[
X_a\perp X_b
\mid X_{\widehat S\setminus\{a,b\}},Z=k
\quad
\Longleftrightarrow
\quad
(\Omega_k)_{ab}=0.
\]
Thus the subgroup-specific graph is
\[
G_k=(\widehat S,E_k),
\]
where
\[
(a,b)\in E_k
\quad
\Longleftrightarrow
\quad
(\Omega_k)_{ab}\ne0.
\]

A high-dimensional estimator of \(\Omega_k\) may be obtained using a penalized likelihood,
\[
\widehat\Omega_k
=
\argmin_{\Omega\succ0}
\left\{
\tr(\widehat\Sigma_k\Omega)
-
\log\det(\Omega)
+
\lambda_k\sum_{a\ne b}|\Omega_{ab}|
\right\},
\]
where \(\widehat\Sigma_k\) is the sample covariance matrix within subgroup \(k\).

For inference on individual edges, use a debiased estimator \(\widehat\Omega_{k,ab}^{d}\). The test statistic is
\[
T_{k,ab}
=
\frac{
\sqrt{\widehat n_k}\widehat\Omega_{k,ab}^{d}
}{
\widehat\sigma_{k,ab}
}.
\]
Under the null hypothesis
\[
H_{0,k,ab}:(\Omega_k)_{ab}=0,
\]
the statistic satisfies
\[
T_{k,ab}\rightsquigarrow N(0,1).
\]
The resulting \(p\)-values are adjusted within each subgroup using the Benjamini-Hochberg procedure. This yields a subgroup-specific dependence graph with controlled false discoveries.

\section{Theoretical Properties}
\label{sec:theory}

The following results are stated for a clean Gaussian version of the model. This is the natural first theory for the proposed method. Extensions to generalized responses or non-Gaussian covariates require the same structure: identifiability, sparsity, regular priors, and concentration of likelihood ratios.

\subsection{Assumptions}
\label{subsec:assumptions}

\begin{assumption}[Sparse truth]
\label{ass:sparse-truth}
There exists a true subset \(S_0\subset\{1,\ldots,p\}\) with
\[
|S_0|=s_0,
\qquad
s_0\log p=o(n).
\]
The true inclusion vector is \(\gamma_0\), where \((\gamma_0)_j=1\) if and only if \(j\in S_0\).
\end{assumption}
This condition formalizes the high-dimensional regime in which only a small number of covariates drive the mixture structure and the subgroup-specific response model. The requirement \(s_0\log p=o(n)\) is the usual effective-sample-size condition for recovering a sparse subset among \(p\) candidates. It links the distillation step to the central statistical goal of the paper: clustering should be performed on the variables that carry subgroup information, not on the full ambient feature space.

\begin{assumption}[Irrelevant variables contain no subgroup information]
\label{ass:irrelevance}
For all \(k\),
\[
p_0(X_{S_0^c}\mid X_{S_0},Z=k)
=
p_0(X_{S_0^c}\mid X_{S_0}).
\]
Thus \(X_{S_0^c}\) may be correlated with \(X_{S_0}\), but it carries no additional information about \(Z\).
\end{assumption}

This assumption gives a precise meaning to irrelevance in the mixture setting. Variables in \(S_0^c\) are allowed to be correlated with the selected variables, so the condition does not impose marginal independence. Rather, it requires that, after conditioning on \(X_{S_0}\), the remaining variables do not further distinguish the latent subgroups. This is the population-level justification for distillation: removing \(S_0^c\) should reduce noise without removing information about the latent class.

\begin{assumption}[Identifiability]
\label{ass:identifiability}
For \(k\ne\ell\),
\[
(\mu_{k0},\Sigma_{k0},\alpha_k,\beta_{k0},\phi_k)
\ne
(\mu_{\ell0},\Sigma_{\ell0},\alpha_\ell,\beta_{\ell0},\phi_\ell).
\]
The mixture is identifiable up to relabeling of the \(K\) components.
\end{assumption}

The latent classes must correspond to distinct data-generating mechanisms. Since mixture labels are inherently arbitrary, identifiability is required only up to permutation of the component labels. Without this condition, neither posterior concentration nor clustering consistency has a well-defined target. The assumption ensures that the inferred subgroups can be compared to a population-level mixture decomposition rather than to an artifact of parameterization.

\begin{assumption}[Eigenvalue control]
\label{ass:eigenvalue}
There exist constants \(0<c<C<\infty\) such that, for every \(k\),
\[
c\le \lambda_{\min}(\Sigma_{k0})
\le
\lambda_{\max}(\Sigma_{k0})
\le C.
\]
\end{assumption}

This condition rules out degenerate or nearly singular within-component covariance matrices. It prevents subgroup separation from being created by ill-conditioned covariance structure rather than by stable distributional differences. The bounded-eigenvalue condition is also needed for likelihood concentration, Gaussian tail bounds, and precision-matrix inference within subgroups.

\begin{assumption}[Signal strength]
\label{ass:signal}
For every \(j\in S_0\), the variable has detectable subgroup or response signal. Define
\[
b_j
=
\max_{k\ne\ell}|\mu_{kj,0}-\mu_{\ell j,0}|
+
\max_k|\beta_{kj,0}|.
\]
Assume
\[
\min_{j\in S_0} b_j \ge b_n,
\]
where
\[
n b_n^2\ge C_1\log p
\]
for a sufficiently large constant \(C_1\).
\end{assumption}

Each relevant variable must leave a detectable trace either through subgroup separation or through the subgroup-specific response model. The lower bound \(n b_n^2\ge C_1\log p\) is a beta-min condition adapted to the high-dimensional setting. It separates true variables from noise at the scale needed to search over many candidate covariates. This condition is what makes posterior distillation possible rather than merely descriptive.

\begin{assumption}[Component separation]
\label{ass:separation}
Define
\[
\Delta_n
=
\min_{k\ne\ell}
\left[
(\mu_{k0}-\mu_{\ell0})^\top
\Sigma_0^{-1}
(\mu_{k0}-\mu_{\ell0})
\right]^{1/2},
\]
where \(\Sigma_0\) is a representative within-component covariance, such as a common covariance in the homoscedastic case. Assume
\[
\Delta_n\to\infty
\]
for vanishing clustering error, and
\[
\Delta_n^2\ge C_2\log n
\]
for exact recovery with probability tending to one.
\end{assumption}

The quantity \(\Delta_n\) measures the minimum separation between latent components in the informative subspace. Consistent clustering cannot be expected when components overlap too strongly, even if the correct variables are known. The stronger logarithmic separation condition gives exact label recovery with high probability. This assumption connects the variable-selection part of the method to the downstream clustering target: after the relevant subspace has been recovered, the components must still be statistically distinguishable.

\begin{assumption}[Prior regularity]
\label{ass:prior}
The prior on \(\gamma\) satisfies
\[
\Prob(\gamma)\propto c^{-|\gamma|}p^{-a|\gamma|},
\qquad a>1,
\]
and the parameter priors have continuous densities that are positive in neighborhoods of the true parameter values.
\end{assumption}

The prior assigns enough mass near the true sparse model while penalizing unnecessarily large models. The factor \(p^{-a|\gamma|}\), with \(a>1\), prevents the posterior from favoring large models that improve the likelihood only by fitting noise. Positivity of the parameter priors near the truth ensures that the Bayesian procedure does not exclude the data-generating parameter values. Together, these conditions make posterior model selection compatible with high-dimensional sparsity.

\subsection{Information Preservation After Distillation}
\label{subsec:information-preservation}

\begin{theorem}[Information preservation after distillation]
\label{thm:information}
Under Assumption~\ref{ass:irrelevance}, the Bayes classifier based on all variables \(X\) is identical to the Bayes classifier based only on \(X_{S_0}\). That is,
\[
\argmax_k \Prob(Z=k\mid X)
=
\argmax_k \Prob(Z=k\mid X_{S_0}).
\]
\end{theorem}

The relevant subspace is sufficient for optimal subgroup classification. Discarding \(X_{S_0^c}\) does not change the population Bayes rule because those variables contain no additional class information once \(X_{S_0}\) is known. This property is central to the paper's construction: distillation is not merely a computational reduction, but a statistically lossless reduction for the clustering target under the stated model.

The proof is given in Appendix~\ref{app:proof-information}.

\subsection{Clustering Consistency}
\label{subsec:clustering-consistency}

\begin{theorem}[Clustering consistency]
\label{thm:clustering}
Assume that \(K\) is fixed, that \(\min_k\pi_{k0}\ge\pi_{\min}>0\), and that Assumptions~\ref{ass:eigenvalue} and \ref{ass:separation} hold with a common within-component covariance \(\Sigma_{k0}=\Sigma_0\) for every \(k\). Let \(\widehat Z_i\) be the label obtained by fitting the Gaussian mixture model, with common or component-specific covariance matrices, on \(X_{i\widehat S}\). Suppose that the fitted parameters are accurate at a rate \(\varepsilon_n\): with probability tending to one, \(\widehat S=S_0\) and, after a suitable relabeling of the components,
\[
\max_{1\le k\le K}
\left(
\|\widehat\mu_k-\mu_{k0}\|_2
+
\|\widehat\Sigma_k-\Sigma_0\|_{\mathrm{op}}
+
|\widehat\pi_k-\pi_{k0}|
\right)
\le\varepsilon_n,
\]
where \(\varepsilon_n\to0\) and \(\varepsilon_n(s_0+\log n)=o(\Delta_n^2)\). Then, up to a permutation of component labels,
\[
\frac1n
\sum_{i=1}^n
\ind(\widehat Z_i\ne Z_i)
=
O_P\{\exp(-c_\ast\Delta_n^2)\}
\]
for every constant \(c_\ast<1/8\).

If, in addition,
\[
\Delta_n^2\ge C_2\log n
\]
for a constant \(C_2>8\), then
\[
\Prob\left(
\min_{\sigma\in\mathfrak S_K}
\max_{1\le i\le n}
\ind\{\widehat Z_i\ne \sigma(Z_i)\}=0
\right)
\to1,
\]
where \(\mathfrak S_K\) is the set of all permutations of \(\{1,\ldots,K\}\).
\end{theorem}

Once the distilled subspace recovers the relevant variables, clustering in that subspace attains the usual Gaussian-mixture classification behavior. The error rate is governed by the separation \(\Delta_n\), not by the ambient dimension \(p\). The accuracy condition on the fitted parameters is an input from the estimation theory of finite Gaussian mixtures, which the argument uses but does not establish. The event \(\widehat S=S_0\) is also part of that condition, since the paper does not establish selection consistency for the rule of Section~\ref{sec:bayesian-distillation}; see Section~\ref{sec:common-space}. This supports the two-stage design: posterior distillation controls the high-dimensional feature problem, while mixture clustering then operates in the statistically relevant coordinates.

The proof is given in Appendix~\ref{app:proof-clustering}.

\subsection{Posterior Contraction for the Response Model}
\label{subsec:posterior-contraction}

\begin{theorem}[Posterior contraction for the response model]
\label{thm:contraction}
Suppose the response model is a regular exponential-family regression model, Assumptions~\ref{ass:sparse-truth}, \ref{ass:identifiability}, \ref{ass:signal}, and \ref{ass:prior} hold, and the posterior concentrates on the true model, in the sense that \(\Pi(\gamma=\gamma_0\mid D_n)\to1\) in \(P_0\)-probability. Then, after relabeling the mixture components,
\[
\Pi\left(
\max_{1\le k\le K}
\|\beta_{k,S_0}-\beta_{k0,S_0}\|_2
>
M\sqrt{\frac{s_0\log p}{n}}
\ \middle|\ D_n
\right)
\to0
\]
in \(P_0\)-probability for a sufficiently large constant \(M\). The same rate holds for \(\mu_{k,S_0}\) under the Gaussian covariate model.
\end{theorem}

After model selection, the subgroup-specific response parameters concentrate at the standard sparse high-dimensional rate. The rate depends on \(s_0\log p\), rather than on the full dimension \(p\). This result links clustering to inference: the distilled variables are not only useful for label recovery, but also support stable estimation of subgroup-specific response relationships.

The proof is given in Appendix~\ref{app:proof-contraction}.

\subsection{Validity of Subgroup Conditional-Independence Tests}
\label{subsec:ci-tests}

\begin{theorem}[Validity of subgroup conditional-independence tests]
\label{thm:ci-tests}
Assume the subgroup-specific covariate distribution is Gaussian. For subgroup \(k\), let
\[
\Omega_{k0}=\Sigma_{k0}^{-1}.
\]
Assume \(\Omega_{k0}\) is sparse, with maximum node degree \(d_k\), and
\[
d_k^2\frac{\log s_0}{n_k}\to0,
\]
where
\[
n_k=\#\{i:Z_i=k\}.
\]
Assume also that Theorem~\ref{thm:clustering} gives exact recovery of the subgroup labels with probability tending to one.

Then, for every null edge \((a,b)\) satisfying
\[
(\Omega_{k0})_{ab}=0,
\]
the debiased statistic satisfies
\[
T_{k,ab}\rightsquigarrow N(0,1).
\]
Consequently, if the Benjamini-Hochberg procedure is applied within subgroup \(k\), the false-discovery rate satisfies
\[
\FDR_k\le q+o(1).
\]
Moreover, if for a nonzero edge
\[
\frac{
\sqrt{n_k}|(\Omega_{k0})_{ab}|
}{
\sigma_{k,ab}
}
\to\infty,
\]
then the probability of detecting that edge tends to one.
\end{theorem}

The subgroup-specific graphical analysis remains valid after clustering, provided the labels are recovered with high probability and the precision matrices are sufficiently sparse. The result separates two sources of uncertainty: label uncertainty is controlled by the clustering theorem, and edge-level uncertainty is handled by debiased precision-matrix inference. This gives a formal basis for using the inferred subgroups to study conditional dependence structures among the distilled variables.

The proof is given in Appendix~\ref{app:proof-ci-tests}.

\section{Practical Methodology}
\label{sec:methodology}

\subsection{Step 1: Preprocessing}
\label{subsec:preprocessing}

Each covariate is centered and scaled. For biological data, variables with extreme missingness or near-zero variance are removed before modeling. Missing values may be handled by model-based imputation, but the imputation model should be fit inside the training folds when cross-validation is used.

\subsection{Step 2: Bayesian Variable Distillation}
\label{subsec:variable-distillation}

Fit the Bayesian mixture-response model with sparse inclusion vector \(\gamma\). Compute posterior inclusion probabilities
\[
w_j=\Pi(\gamma_j=1\mid D_n).
\]
Select
\[
\widehat S=\{j:w_j\ge \widehat t\}
\]
using posterior false-discovery control.

This step is not a marginal screening rule. A variable is selected because the posterior evidence supports its role in subgroup separation, subgroup-specific response behavior, or both.

\subsection{Step 3: Refit the Mixture on the Distilled Subspace}
\label{subsec:refit-mixture}

Fit a finite mixture model to
\[
\{X_{i\widehat S}:1\le i\le n\}.
\]
Compute posterior membership weights \(\widehat r_{ik}\) and assign
\[
\widehat Z_i=\argmax_k\widehat r_{ik}.
\]

Uncertainty should be reported through the entropy
\[
H_i
=
-\sum_{k=1}^K
\widehat r_{ik}\log \widehat r_{ik}.
\]
Subjects with large \(H_i\) are weakly assigned and should not be overinterpreted.

\subsection{Step 4: Estimate Subgroup-Specific Response Models}
\label{subsec:response-models}

Within each inferred subgroup, fit
\[
Y_i\mid X_{i\widehat S},\widehat Z_i=k.
\]
This provides subgroup-specific associations between the selected covariates and the response. These estimates are useful for interpretation, prediction, and treatment-effect heterogeneity analysis.

\subsection{Step 5: Conditional-Independence Diagnostics}
\label{subsec:practical-ci}

For each subgroup \(k\), estimate the precision matrix \(\Omega_k\) of \(X_{\widehat S}\). Test
\[
H_{0,k,ab}:(\Omega_k)_{ab}=0
\]
for all pairs \(a\ne b\). Adjust the resulting \(p\)-values using Benjamini-Hochberg within each subgroup.

The output is a collection of graphs
\[
\widehat G_1,\ldots,\widehat G_K,
\]
where each graph describes conditional dependence among selected variables inside one subgroup.

\subsection{Step 6: Model Checking}
\label{subsec:model-checking}

The following checks should be reported:
\begin{enumerate}[label=\arabic*.]
    \item Stability of \(\widehat S\) under resampling.
    \item Stability of cluster assignments under resampling.
    \item Posterior uncertainty of \(K\), if \(K\) is not fixed.
    \item Predictive performance for \(Y\).
    \item Separation of clusters in the selected subspace.
    \item Sensitivity to the posterior false-discovery level \(q\).
    \item Sensitivity to the prior on model size.
\end{enumerate}

\section{Simulation Studies}
\label{sec:simulation}

This section reports a simulation study of the full procedure, with the selection step carried out under the full-data posterior \eqref{eq:full-posterior} of Section~\ref{sec:common-space}. The posterior as written in Section~\ref{sec:bayesian-distillation} selects no coordinates in comparable settings (Section~\ref{subsec:common-space-sim}), so it is not included.

\subsection{Design}
\label{subsec:sim-design}

The design crosses \(n\in\{200,500\}\), \(p\in\{500,2000\}\), \(s_0\in\{5,10\}\) and \(K\in\{2,3\}\), giving \(16\) settings, each replicated \(50\) times. The labels are drawn with \(\Prob(Z_i=k)=1/K\). Given \(Z_i=k\), the relevant coordinates satisfy \(X_{iS_0}\sim N_{s_0}(\mu_k,\Sigma_0)\), where \(\Sigma_0\) has entries \(0.3^{|a-b|}\). For each relevant coordinate, the subgroup means are a random permutation of \(\{-1.25,1.25\}\) when \(K=2\) and of \(\{-1.25,0,1.25\}\) when \(K=3\), drawn once per replication. Twenty irrelevant coordinates are noisy copies \(0.5X_{ir}+\varepsilon_{ij}\) of a randomly chosen relevant coordinate \(r\), so they are correlated with \(X_{iS_0}\), differ in mean across subgroups, and still satisfy Assumption~\ref{ass:irrelevance}; the remaining coordinates are independent standard normal. The response is
\[
Y_i=\alpha_{Z_i}+X_{iS_0}^\top\beta_{Z_i}+\varepsilon_i,
\qquad
\varepsilon_i\sim N(0,1),
\]
with \((\alpha_1,\alpha_2,\alpha_3)=(0,2,-2)\) and \(\beta_{kj}=0.5(-1)^{k+j}\), so that the response surfaces differ across subgroups. All covariates are standardized as in Section~\ref{subsec:preprocessing}, and prediction error is measured on an independent test sample of size \(500\). The number of subgroups \(K\) is treated as known by every method.

\subsection{Methods}
\label{subsec:sim-methods}

The proposed procedure follows Section~\ref{sec:methodology}. The selection step uses the large-sample form of \eqref{eq:full-posterior}: each inclusion vector is scored by the penalized maximized log-likelihood of a \(K\)-component Gaussian mixture with common covariance on \(X_{S_\gamma}\), minus that of a single Gaussian on the same coordinates, plus the logarithm of the complexity prior with \(a=1.1\) and \(c=1\). A common covariance keeps the maximized likelihood finite (Section~\ref{subsec:units}) and matches the setting of Theorem~\ref{thm:clustering}. The posterior is explored by a Metropolis--Hastings sampler with addition, deletion and swap moves, run for \(3000\) iterations after which the first \(1000\) are discarded. For computational reasons the sampler moves within a candidate pool of at most \(50\) coordinates, formed by the \(25\) coordinates with the largest one-dimensional mixture-versus-Gaussian criterion and the \(25\) with the largest loadings on the leading \(K-1\) principal components. It starts from the ten highest-loading coordinates, with subgroups initialized at the spectral solution described below. Posterior inclusion probabilities \(w_j\) are the visit frequencies of the sampler, and \(\widehat S\) is obtained from the false-discovery rule of Section~\ref{sec:bayesian-distillation} with \(q=0.1\). The mixture is then refitted on \(X_{\widehat S}\), least squares is fitted within each estimated subgroup for prediction, and conditional-independence edges among the selected coordinates are tested within each estimated subgroup. Since \(|\widehat S|\) is small relative to the subgroup sizes in these settings, the edge tests use classical partial correlations with Fisher's transformation, which play the role of the debiased statistics of Section~\ref{sec:conditional-independence}; the Benjamini--Hochberg procedure is applied within each subgroup at level \(0.1\). The candidate pool, the starting point and the number of iterations were fixed on pilot runs with random seeds different from those of the reported replications.

The comparison methods are those listed in the original design of this section. The full-dimensional Gaussian mixture uses diagonal covariance matrices, since unrestricted covariances cannot be estimated when \(p>n\). Sparse \(K\)-means \citep{witten2010framework} selects the coordinates with nonzero weight, with the \(\ell_1\) bound chosen by the gap statistic over a grid from \(1.1\) to \(\sqrt p\), as recommended by its authors. The penalized mixture is the \(\ell_1\)-penalized Gaussian mixture with diagonal common covariance of \citet{pan2007penalized}, with the penalty chosen by the Bayesian information criterion and selected coordinates those with a nonzero mean in some component. Spectral clustering applies \(K\)-means to the leading \(K\) principal component scores of all coordinates. Clustering after marginal screening fits the common-covariance mixture to the coordinates whose one-dimensional criterion is positive. An oracle fits the same mixture to \(X_{S_0}\), and for prediction it uses the true labels and \(S_0\). A pooled lasso, which ignores the subgroups, provides a prediction benchmark.

Clustering accuracy is measured by the adjusted Rand index and by the misclassification rate after the best matching of labels. Selection is measured by the proportion of \(S_0\) selected (TPR) and by the proportion of selected coordinates outside \(S_0\) (FDR). For the edge tests, the true graph on \(\widehat S\) is computed from the exact within-subgroup covariance of \(X_{\widehat S}\), which is the same in every subgroup in this design.

\subsection{Results}
\label{subsec:sim-results}

Table~\ref{tab:sim-selection} and Figure~\ref{fig:sim-selection} report the selection results. The proposed procedure is the only one whose selected sets contain almost no irrelevant coordinates: its average false-discovery proportion never exceeds \(0.011\). Sparse \(K\)-means and the penalized mixture find nearly every relevant coordinate, with true positive rates of at least \(0.98\) and \(0.88\), respectively, but sparse \(K\)-means gives positive weight to \(1137\) coordinates on average and the penalized mixture selects about \(39\), so that their false-discovery proportions are at least \(0.97\) and between \(0.67\) and \(0.88\), respectively. Marginal screening works when \(K=2\) but finds almost nothing when \(K=3\), because a coordinate whose three subgroup means are equally spaced is only weakly multimodal on its own. The price of the proposed procedure's conservatism is power. Its true positive rate is between \(0.84\) and \(0.98\) when \(n=500\), but falls to between \(0.37\) and \(0.82\) when \(n=200\) and \(K=3\), where the penalty of the complexity prior on each additional coordinate outweighs the evidence in \(200\) observations.

\begin{table}[!htbp]
\centering
\small
\caption{Variable selection: proportion of \(S_0\) selected (TPR) and false-discovery proportion (FDR), averaged over \(50\) replications, for Bayesian distillation (Proposed), sparse \(K\)-means, the penalized mixture and marginal screening. Monte Carlo standard errors are at most 0.06.}
\label{tab:sim-selection}
\medskip
\begin{tabular}{rrrr|cc|cc|cc|cc}
\hline
 & & & & \multicolumn{2}{c|}{Proposed} & \multicolumn{2}{c|}{Sparse \(K\)-means} & \multicolumn{2}{c|}{Penalized} & \multicolumn{2}{c}{Screening} \\
\(n\) & \(p\) & \(s_0\) & \(K\) & TPR & FDR & TPR & FDR & TPR & FDR & TPR & FDR \\
\hline
200 & 500 & 5 & 2 & 0.77 & 0.00 & 1.00 & 0.99 & 1.00 & 0.82 & 0.61 & 0.41 \\
200 & 500 & 5 & 3 & 0.60 & 0.01 & 1.00 & 0.99 & 1.00 & 0.82 & 0.00 & 0.06 \\
200 & 500 & 10 & 2 & 0.66 & 0.00 & 1.00 & 0.98 & 1.00 & 0.69 & 0.66 & 0.24 \\
200 & 500 & 10 & 3 & 0.82 & 0.00 & 1.00 & 0.98 & 1.00 & 0.71 & 0.00 & 0.06 \\
500 & 500 & 5 & 2 & 0.92 & 0.00 & 1.00 & 0.97 & 1.00 & 0.87 & 0.97 & 0.10 \\
500 & 500 & 5 & 3 & 0.93 & 0.00 & 1.00 & 0.99 & 1.00 & 0.88 & 0.01 & 0.00 \\
500 & 500 & 10 & 2 & 0.86 & 0.01 & 1.00 & 0.98 & 1.00 & 0.76 & 0.99 & 0.06 \\
500 & 500 & 10 & 3 & 0.98 & 0.00 & 1.00 & 0.98 & 1.00 & 0.81 & 0.01 & 0.00 \\
200 & 2000 & 5 & 2 & 0.76 & 0.00 & 1.00 & 1.00 & 1.00 & 0.87 & 0.62 & 0.72 \\
200 & 2000 & 5 & 3 & 0.37 & 0.00 & 0.98 & 1.00 & 0.88 & 0.77 & 0.00 & 0.12 \\
200 & 2000 & 10 & 2 & 0.61 & 0.00 & 1.00 & 0.99 & 1.00 & 0.76 & 0.63 & 0.60 \\
200 & 2000 & 10 & 3 & 0.57 & 0.00 & 1.00 & 0.98 & 1.00 & 0.80 & 0.00 & 0.12 \\
500 & 2000 & 5 & 2 & 0.95 & 0.00 & 1.00 & 1.00 & 1.00 & 0.80 & 0.97 & 0.36 \\
500 & 2000 & 5 & 3 & 0.94 & 0.00 & 1.00 & 1.00 & 1.00 & 0.88 & 0.00 & 0.02 \\
500 & 2000 & 10 & 2 & 0.84 & 0.00 & 1.00 & 0.99 & 1.00 & 0.67 & 0.98 & 0.21 \\
500 & 2000 & 10 & 3 & 0.98 & 0.00 & 1.00 & 0.99 & 1.00 & 0.70 & 0.00 & 0.00 \\
\hline
\end{tabular}
\end{table}

\begin{figure}[!htbp]
\centering
\includegraphics[width=\linewidth]{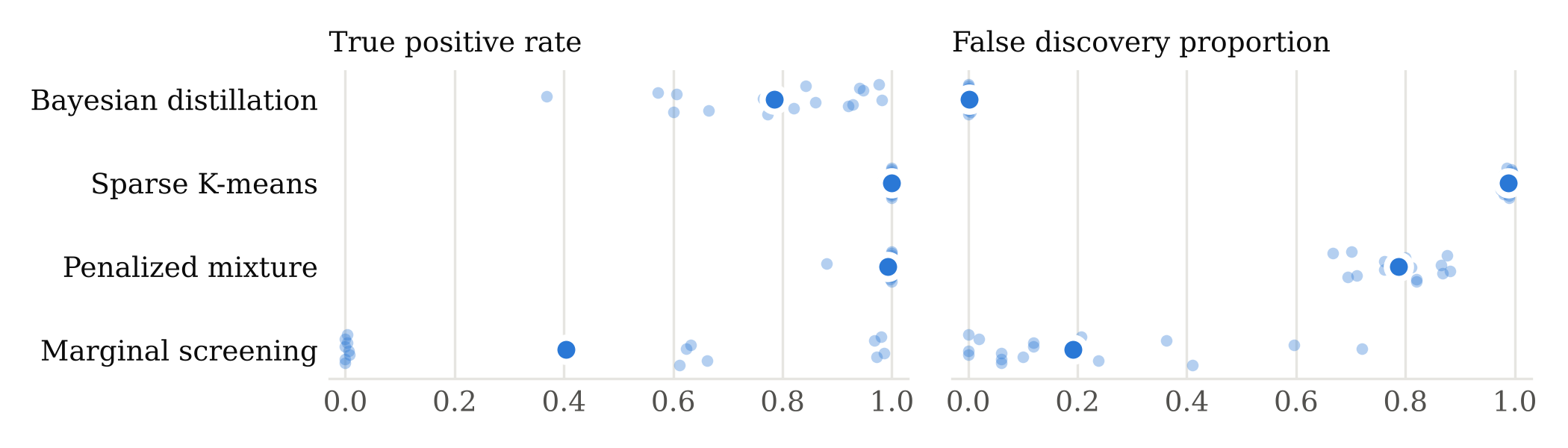}
\caption{Variable selection by the four methods that select coordinates. Small points are the averages in each of the \(16\) settings; large points are the averages over settings.}
\label{fig:sim-selection}
\end{figure}

Table~\ref{tab:sim-ari} and Figure~\ref{fig:sim-clustering} report clustering accuracy. With \(n=500\), the proposed procedure has the highest or joint-highest adjusted Rand index in all eight settings. Its advantage is clearest with three subgroups and five relevant coordinates, where it attains \(0.82\) against at most \(0.76\) for the other methods and \(0.85\) to \(0.86\) for the oracle. With \(n=200\) and two subgroups, the three variable-selecting mixture methods are indistinguishable, with indices between \(0.97\) and \(1.00\). With \(n=200\) and three subgroups, the proposed procedure is worse in three of the four settings, by up to \(0.20\) (\(0.76\) against \(0.95\) for the penalized mixture with \(p=2000\) and \(s_0=10\)), because it selects too few coordinates. Averaged over all settings, the adjusted Rand index is \(0.89\) for the proposed procedure, \(0.89\) for sparse \(K\)-means and \(0.90\) for the penalized mixture, against \(0.95\) for the oracle; spectral clustering (\(0.72\)), clustering after marginal screening (\(0.49\)) and the full-dimensional mixture (\(0.40\)) are clearly worse. The proposed procedure therefore matches the best competitors in recovering the subgroups while returning a selected set that identifies the variables defining them, which the competitors do not provide.

\begin{table}[!htbp]
\centering
\small
\caption{Adjusted Rand index, averaged over \(50\) replications. BD: Bayesian distillation; SKM: sparse \(K\)-means; PM: penalized mixture; MS: marginal screening; SC: spectral clustering; FM: full-dimensional mixture. The largest value among the non-oracle methods is in bold. Monte Carlo standard errors are at most 0.07.}
\label{tab:sim-ari}
\medskip
\begin{tabular}{rrrr|cccccc|c}
\hline
\(n\) & \(p\) & \(s_0\) & \(K\) & BD & SKM & PM & MS & SC & FM & Oracle \\
\hline
200 & 500 & 5 & 2 & \textbf{0.97} & \textbf{0.97} & \textbf{0.97} & 0.91 & 0.92 & 0.85 & 0.99 \\
200 & 500 & 5 & 3 & 0.60 & \textbf{0.70} & \textbf{0.70} & 0.00 & 0.44 & 0.11 & 0.84 \\
200 & 500 & 10 & 2 & \textbf{1.00} & \textbf{1.00} & \textbf{1.00} & \textbf{1.00} & 0.99 & 0.99 & 1.00 \\
200 & 500 & 10 & 3 & \textbf{0.96} & 0.94 & 0.94 & 0.00 & 0.74 & 0.21 & 0.98 \\
500 & 500 & 5 & 2 & \textbf{0.99} & 0.98 & 0.97 & \textbf{0.99} & 0.94 & 0.96 & 0.99 \\
500 & 500 & 5 & 3 & \textbf{0.82} & 0.75 & 0.75 & 0.01 & 0.57 & 0.45 & 0.85 \\
500 & 500 & 10 & 2 & \textbf{1.00} & \textbf{1.00} & \textbf{1.00} & \textbf{1.00} & \textbf{1.00} & \textbf{1.00} & 1.00 \\
500 & 500 & 10 & 3 & \textbf{0.98} & 0.96 & 0.96 & 0.02 & 0.86 & 0.81 & 0.98 \\
200 & 2000 & 5 & 2 & \textbf{0.98} & 0.97 & \textbf{0.98} & 0.87 & 0.80 & 0.02 & 0.99 \\
200 & 2000 & 5 & 3 & 0.37 & 0.47 & \textbf{0.53} & 0.00 & 0.17 & 0.01 & 0.84 \\
200 & 2000 & 10 & 2 & \textbf{1.00} & \textbf{1.00} & \textbf{1.00} & 0.98 & 0.95 & 0.05 & 1.00 \\
200 & 2000 & 10 & 3 & 0.76 & 0.87 & \textbf{0.95} & 0.01 & 0.30 & 0.01 & 0.98 \\
500 & 2000 & 5 & 2 & \textbf{0.99} & 0.98 & 0.98 & \textbf{0.99} & 0.90 & 0.19 & 0.99 \\
500 & 2000 & 5 & 3 & \textbf{0.82} & 0.73 & 0.76 & 0.00 & 0.34 & 0.01 & 0.86 \\
500 & 2000 & 10 & 2 & \textbf{1.00} & \textbf{1.00} & \textbf{1.00} & \textbf{1.00} & 0.99 & 0.67 & 1.00 \\
500 & 2000 & 10 & 3 & \textbf{0.99} & 0.96 & 0.96 & 0.01 & 0.62 & 0.02 & 0.99 \\
\hline
\end{tabular}
\end{table}

\begin{figure}[!htbp]
\centering
\includegraphics[width=\linewidth]{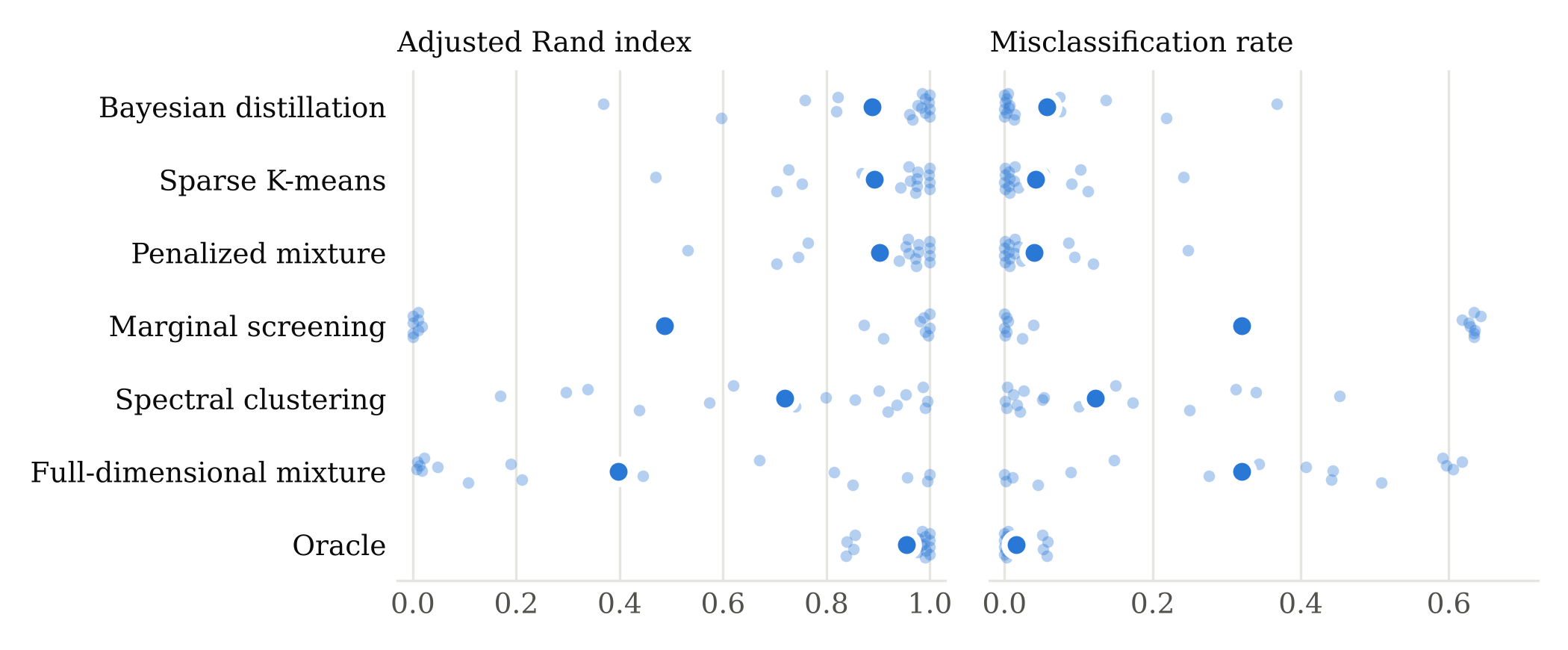}
\caption{Clustering accuracy of all methods. Small points are the averages in each of the \(16\) settings; large points are the averages over settings. The misclassification rate is computed after the best matching of labels.}
\label{fig:sim-clustering}
\end{figure}

Table~\ref{tab:sim-prediction} reports the downstream steps of the proposed procedure. Subgroup-specific least squares on the selected coordinates has lower test error than the pooled lasso in \(15\) of the \(16\) settings, and with \(n=500\) it is within \(0.6\) of the oracle, which uses the true labels and \(S_0\). The exception is \(n=200\), \(p=2000\), \(s_0=5\) and \(K=3\), where only \(37\%\) of the relevant coordinates are selected. The within-subgroup edge tests have false-discovery proportions between \(0.03\) and \(0.24\) and power between \(0.29\) and \(0.92\); with \(n=500\), the power is at least \(0.84\). The false-discovery proportion exceeds the nominal \(0.1\) in four settings, all with \(K=3\), where the estimated subgroups contain the most misclassified observations. The largest value, \(0.24\), occurs in the same setting as the exception above, where the adjusted Rand index is \(0.37\). This agrees with Theorem~\ref{thm:ci-tests}, which assumes exact recovery of the labels: an estimated subgroup that mixes observations from different true subgroups induces conditional dependence that is not present within any true subgroup.

\begin{table}[!htbp]
\centering
\small
\caption{Test mean squared error of subgroup-specific prediction (the noise variance is \(1\)) and within-subgroup edge tests at level \(0.1\), averaged over \(50\) replications. The oracle uses the true labels and \(S_0\); the pooled lasso ignores the subgroups. Monte Carlo standard errors are at most 0.25 for the prediction errors and 0.06 for the edge tests.}
\label{tab:sim-prediction}
\medskip
\begin{tabular}{rrrr|ccc|cc}
\hline
 & & & & \multicolumn{3}{c|}{Test error} & \multicolumn{2}{c}{Edge tests} \\
\(n\) & \(p\) & \(s_0\) & \(K\) & Proposed & Oracle & Pooled lasso & FDR & Power \\
\hline
200 & 500 & 5 & 2 & 1.31 & 1.05 & 2.03 & 0.03 & 0.62 \\
200 & 500 & 5 & 3 & 2.70 & 1.10 & 3.06 & 0.10 & 0.52 \\
200 & 500 & 10 & 2 & 1.72 & 1.10 & 2.69 & 0.06 & 0.56 \\
200 & 500 & 10 & 3 & 1.77 & 1.19 & 3.33 & 0.14 & 0.30 \\
500 & 500 & 5 & 2 & 1.11 & 1.02 & 1.90 & 0.06 & 0.90 \\
500 & 500 & 5 & 3 & 1.62 & 1.03 & 2.74 & 0.12 & 0.89 \\
500 & 500 & 10 & 2 & 1.34 & 1.05 & 2.58 & 0.08 & 0.86 \\
500 & 500 & 10 & 3 & 1.18 & 1.06 & 3.11 & 0.09 & 0.86 \\
200 & 2000 & 5 & 2 & 1.34 & 1.09 & 2.07 & 0.06 & 0.69 \\
200 & 2000 & 5 & 3 & 3.47 & 1.10 & 2.97 & 0.24 & 0.46 \\
200 & 2000 & 10 & 2 & 1.86 & 1.14 & 2.79 & 0.08 & 0.49 \\
200 & 2000 & 10 & 3 & 2.61 & 1.19 & 3.30 & 0.13 & 0.29 \\
500 & 2000 & 5 & 2 & 1.11 & 1.04 & 1.97 & 0.06 & 0.92 \\
500 & 2000 & 5 & 3 & 1.62 & 1.04 & 2.70 & 0.08 & 0.92 \\
500 & 2000 & 10 & 2 & 1.36 & 1.04 & 2.63 & 0.09 & 0.84 \\
500 & 2000 & 10 & 3 & 1.23 & 1.09 & 3.11 & 0.10 & 0.88 \\
\hline
\end{tabular}
\end{table}

\subsection{Limitations of the study}
\label{subsec:sim-limitations}

The data are generated from a common-covariance Gaussian mixture, which is the working model of the proposed selection criterion, and \(K\) is known to every method. The comparison methods were implemented for this study following their published descriptions, not with their authors' software. The proposed procedure uses the large-sample form of the posterior and a candidate pool rather than an exhaustive sampler over all \(2^p\) inclusion vectors, so the visit frequencies approximate the posterior inclusion probabilities only within that pool. A study with misspecified component shapes, an unknown number of subgroups and a real data set would complement these results.

\clearpage
\section{Comparing inclusion vectors on a common sample space}
\label{sec:common-space}

The posterior of Section~\ref{sec:bayesian-distillation} compares inclusion vectors through \(L_n(\Theta,\gamma)\). For each \(\gamma\), this function is a density of \((Y_i,X_{iS_\gamma})_{i\le n}\) alone. For two inclusion vectors with different selected sets, the two functions are densities of different random vectors with respect to dominating measures of different dimensions, so their ratio is not a likelihood ratio for the observed data \((Y_i,X_i)_{i\le n}\). Formally, \(\Pi(\gamma\mid D_n)\) is the posterior obtained when the excluded coordinates are given the improper constant density one, a choice whose effect depends on the units in which those coordinates are measured. The model of Section~\ref{sec:setting} contains the factor \(h(x_{S^c}\mid x_S)\) for the excluded coordinates; the likelihood of Section~\ref{sec:bayesian-distillation} drops it. This section shows that the omission changes the selected set in a way that has nothing to do with subgroup structure, derives the posterior that results when the factor is restored, and illustrates the difference numerically.

\subsection{Units determine the selected set}
\label{subsec:units}

Write
\[
\ell_n(\gamma)=\sup_{\Theta\in\mathcal T_\gamma}\log L_n(\Theta,\gamma),
\]
where \(\mathcal T_\gamma\) is the parameter set of the Gaussian covariate model with unrestricted component means and positive definite component covariance matrices, together with a response model in which \(X_{iS_\gamma}\) enters only through the linear predictor \(\eta_{ik}=\alpha_k+X_{iS_\gamma}^\top\beta_k\) with unrestricted coefficients. The set \(\mathcal T_\gamma\) may also carry any constraint that is preserved by the rescaling map used in the proof, such as \(\lambda_{\min}(\Sigma_k\Sigma_\ell^{-1})\ge\rho\) for all \(k,\ell\), with \(\rho\in(0,1]\) \citep{hathaway1985constrained}, together with a lower bound on the ratios of the dispersion parameters of the response. For \(c\neq0\) and a coordinate \(j\), let \(\ell_n^{(c)}(\gamma)\) denote the same quantity computed after column \(j\) of the design has been multiplied by \(c\).

\begin{lemma}[Exact dependence on the units of a coordinate]
\label{lem:units}
For every inclusion vector \(\gamma\), as an identity in \((-\infty,\infty]\),
\[
\ell_n^{(c)}(\gamma)=
\begin{cases}
\ell_n(\gamma)-n\log|c|, & j\in S_\gamma,\\
\ell_n(\gamma), & j\notin S_\gamma.
\end{cases}
\]
\end{lemma}

The identity has content only when \(\ell_n\) is finite. Without a constraint, and with \(K\ge2\), the supremum is \(+\infty\) for every nonempty \(S_\gamma\), because a component can be centred at one observation with covariance tending to zero; a Gaussian response with free dispersion behaves in the same way. The discussion below therefore takes \(\mathcal T_\gamma\) to carry a constraint of the kind just described and assumes that it makes \(\ell_n\) finite. A common covariance matrix is the case \(\rho=1\).

Let a criterion compare inclusion vectors through \(\ell_n(\gamma)\) plus a penalty that depends on \(\gamma\) only through \(|\gamma|\), as do the Bayesian information criterion and the complexity prior of Section~\ref{sec:bayesian-distillation}. By Lemma~\ref{lem:units}, the comparison of \(\gamma\cup\{j\}\) with \(\gamma\) changes by exactly \(-n\log|c|\) when coordinate \(j\) is rescaled by \(c\). Measuring \(X_j\) in sufficiently small units therefore forces \(j\) into the selected set, and measuring it in sufficiently large units forces it out, for every sample size and whatever the role of \(j\) in the mixture. Under the regularity conditions that justify the Laplace approximation, with \(\ell_n\) computed under such a constraint, \(\log m_\gamma(D_n)=\ell_n(\gamma)+O_P(\log n)\) for each fixed \(\gamma\), so the marginal likelihoods of Section~\ref{sec:bayesian-distillation} inherit the same shift up to terms of lower order; the shift is not exact there only because the priors of Section~\ref{sec:bayesian-distillation} are not scale equivariant.

\begin{remark}[Standardization does not remove the problem]
\label{rem:entropy}
Standardization, recommended in Section~\ref{subsec:preprocessing}, fixes the units but not the criterion. After standardization, \(\ell_n(\gamma\cup\{j\})-\ell_n(\gamma)\) is at most the maximized log conditional density of \(X_j\) given \(X_{S_\gamma}\) under the fitted mixture. For every conditional density \(f\), Gibbs' inequality gives \(\E_0\log f(X_j\mid X_{S})\leq-h_0(X_j\mid X_S)\), where \(h_0\) denotes conditional differential entropy under the true law. The population counterpart of the per-observation increase is therefore at most \(-h_0(X_j\mid X_S)\). Consider a standardized coordinate with \(X_j\mid Z\sim N(\pm\delta_j\sigma_j/2,\sigma_j^2)\), where \(\sigma_j^2=1/(1+\delta_j^2/4)\), so that \(X_j\) has unit variance and the two subgroups are \(\delta_j\) within-group standard deviations apart. Since conditioning cannot increase entropy, \(h_0(X_j)\geq h_0(X_j\mid Z)=\tfrac12\log(2\pi e\sigma_j^2)\), which is positive whenever \(\delta_j<8.02\). Computing \(h_0(X_j)\) for the two-component mixture directly shows that it stays positive until \(\delta_j\approx16.4\). Adding such a coordinate to the empty set thus lowers the criterion at a rate linear in \(n\) unless the two subgroups are more than about sixteen within-group standard deviations apart along that coordinate. In the other direction, a coordinate that is almost an exact linear function of coordinates already selected has negative conditional entropy, and adding it can raise the criterion although it carries no subgroup information. The criterion rewards predictability of \(X_j\) from \(X_S\), not information about \(Z\).
\end{remark}

\subsection{The posterior under the full-data model}
\label{subsec:full-data}

Restoring the factor \(h\) of Section~\ref{sec:setting} puts every inclusion vector on the same sample space. Under inclusion vector \(\gamma\) with selected set \(S=S_\gamma\), let the observations be independent with density
\begin{equation}
\label{eq:full-data-model}
p_\gamma(y,x)
=
p^{\mathrm{mix}}_\gamma(y,x_S;\Theta)\,
g_\gamma(x_{S^c}\mid x_S;\vartheta_\gamma),
\end{equation}
where \(p^{\mathrm{mix}}_\gamma(y,x_S;\Theta)=\sum_{k=1}^K\pi_kq_k(x_S;\psi_k)r_k(y\mid x_S;\theta_k)\) is the factor appearing in \(L_n(\Theta,\gamma)\), and \(g_\gamma(\cdot\mid x_S;\vartheta_\gamma)\) is the conditional density of the coordinates in \(S^c\) given those in \(S\) under a \(p\)-variate Gaussian law \(N_p(m,V)\). The parameter \(\vartheta_\gamma=(a_\gamma,B_\gamma,V_{S^c\mid S})\) consists of the intercept \(a_\gamma=m_{S^c}-B_\gamma m_S\), the coefficient matrix \(B_\gamma=V_{S^cS}V_{SS}^{-1}\), and the residual covariance \(V_{S^c\mid S}=V_{S^cS^c}-V_{S^cS}V_{SS}^{-1}V_{SS^c}\). Model \eqref{eq:full-data-model} is the display of Section~\ref{sec:setting} with a Gaussian choice of \(h\); because \(g_\gamma\) does not depend on \(k\), it satisfies Assumption~\ref{ass:irrelevance} with \(S_0\) replaced by \(S\).

Given \(\gamma\), let \(\Theta\) have the prior of Section~\ref{sec:bayesian-distillation}. Independently of \(\Theta\), let
\begin{equation}
\label{eq:niw-prior}
V\sim\operatorname{IW}_p(\Psi_0,\nu_0),
\qquad
m\mid V\sim N_p(m_0,V/\kappa_0),
\qquad
\nu_0>p-1,
\qquad
\Psi_0\succ0,
\qquad
\kappa_0>0,
\end{equation}
and give \(\vartheta_\gamma\) the distribution it inherits from \((m,V)\). For a subset \(A\subseteq\{1,\ldots,p\}\), let \(m^{G}(X_A)\) denote the marginal likelihood of the \(n\times|A|\) matrix \(X_A\) under independent \(N_{|A|}(m_A,V_{AA})\) observations, with the prior on \((m_A,V_{AA})\) inherited from \eqref{eq:niw-prior}, and set \(m^G(X_\emptyset)=1\). Write \(m^{\mathrm{mix}}_\gamma(D_n)=\int L_n(\Theta,\gamma)p(\Theta\mid\gamma)\,d\Theta\) for the marginal likelihood used in Section~\ref{sec:bayesian-distillation}.

\begin{theorem}[Posterior under the full-data model]
\label{thm:full-data-posterior}
For every inclusion vector \(\gamma\), the marginal likelihood of the full data \(D_n\) under model \eqref{eq:full-data-model} with the priors above is
\begin{equation}
\label{eq:full-marginal}
m^{\mathrm{full}}_\gamma(D_n)
=
m^{\mathrm{mix}}_\gamma(D_n)\,
\frac{m^G(X)}{m^G(X_{S_\gamma})},
\end{equation}
where \(m^G(X)\) is finite and does not depend on \(\gamma\). Consequently,
\begin{equation}
\label{eq:full-posterior}
\Pi^{\mathrm{full}}(\gamma\mid D_n)
=
\frac{
p(\gamma)\,m^{\mathrm{mix}}_\gamma(D_n)/m^G(X_{S_\gamma})
}{
\sum_{\gamma'}p(\gamma')\,m^{\mathrm{mix}}_{\gamma'}(D_n)/m^G(X_{S_{\gamma'}})
}.
\end{equation}
Moreover, \(m^G(X_S)\) is the marginal likelihood of \(n\) independent \(N_s(m_S,V_{SS})\) observations, \(s=|S|\), under the prior \(V_{SS}\sim\operatorname{IW}_s(\Psi_{0,SS},\nu_0-p+s)\) and \(m_S\mid V_{SS}\sim N_s(m_{0,S},V_{SS}/\kappa_0)\), and is available in closed form.
\end{theorem}

The proof is given in Appendix~\ref{app:proof-common-space}. The correction to the posterior of Section~\ref{sec:bayesian-distillation} is a single division: each model's marginal likelihood is divided by the marginal likelihood of a single Gaussian on the same coordinates. Every term in \eqref{eq:full-posterior} involves only the \(s\) selected coordinates, so the posterior remains computable when \(p\gg n\), and the Gibbs and Metropolis moves used for the posterior of Section~\ref{sec:bayesian-distillation} apply unchanged. Taking \(\nu_0=p-1+\nu\) with \(\nu>0\) makes the inherited prior \(\operatorname{IW}_s(\Psi_{0,SS},s-1+\nu)\), which does not depend on \(p\). The posterior inclusion probabilities \(w_j\) and the false-discovery rule of Section~\ref{sec:bayesian-distillation} are then computed from \eqref{eq:full-posterior}.

\begin{corollary}[Invariance of the full-data criterion]
\label{cor:invariance}
Let \(\ell^G_n(S)\) be the maximized log-likelihood of \(X_S\) under a single Gaussian law with unrestricted mean and positive definite covariance, with \(\ell^G_n(\emptyset)=0\). Suppose that \(\ell_n(\gamma)\) is computed under a constraint of the kind described in Section~\ref{subsec:units} and that both terms are finite, which for \(\ell^G_n(S_\gamma)\) holds when \(n>|S_\gamma|\) and the observations are in general position. Then \(\ell_n(\gamma)-\ell^G_n(S_\gamma)\) is unchanged when the covariates are transformed by \(x\mapsto Dx+b\) for any diagonal nonsingular \(D\) and any \(b\in\R^p\).
\end{corollary}

The quantity in Corollary~\ref{cor:invariance} is the maximized-likelihood analogue of \(\log m^{\mathrm{mix}}_\gamma-\log m^G(X_{S_\gamma})\). It compares a \(K\)-component mixture with a single Gaussian on the same coordinates, so it rewards cluster structure rather than predictability. When \(X_j\) given \(X_S\) is Gaussian with a mean linear in \(X_S\) and a law that does not depend on \(k\), which is the Gaussian form of Assumption~\ref{ass:irrelevance}, both models describe \(X_j\) given \(X_S\) equally well, and adding \(j\) changes the criterion only through the penalty.

\begin{remark}[Scope]
\label{rem:common-space-scope}
Modeling the excluded coordinates through a regression that does not depend on the component is the device of \citet{raftery2006variable} and \citet{maugis2009variable}. Theorem~\ref{thm:full-data-posterior} shows that, with the normal--inverse-Wishart reference prior \eqref{eq:niw-prior}, the regression on the \(p-s\) excluded coordinates contributes only the ratio \(m^G(X)/m^G(X_S)\). This rests on the a priori independence of the marginal and conditional parameters of a normal--inverse-Wishart law \citep{geiger2002parameter}, and any reference prior with that property leads to the same form. The Gaussian law for the excluded coordinates is a working model; misspecification of it affects \eqref{eq:full-posterior} only through \(m^G(X_S)\). No selection-consistency result is claimed for the posterior of Section~\ref{sec:bayesian-distillation}: by Lemma~\ref{lem:units}, no bound on the posterior mass of overfitted models can hold for that posterior without restrictions on the units of the excluded coordinates. Establishing selection consistency for \(\Pi^{\mathrm{full}}\) is left for future work.
\end{remark}

\subsection{Numerical illustration}
\label{subsec:common-space-sim}

The simulation compares the two criteria inside the same forward-backward search, so that any difference is due to the criterion. Each data set has \(n=300\) observations, \(K=2\) equally likely subgroups, and \(p\in\{50,200\}\) covariates. The five relevant coordinates satisfy \(X_j\mid Z=k\sim N(\pm1.25,1)\). The next five are near copies, \(X_{j+5}=X_j+0.2\,\varepsilon_j\). Twenty coordinates are noisy linear functions \(0.5X_{r}+\varepsilon\) of a randomly chosen relevant coordinate \(r\), and the rest are independent standard normal. The near copies and the linear coordinates satisfy Assumption~\ref{ass:irrelevance}, so \(S_0\) consists of the first five coordinates. All columns are standardized, as in Section~\ref{subsec:preprocessing}. The response is omitted, so both criteria use only the covariate part of the model.

The criterion as written is \(\ell_n(\gamma)\), with \(\ell_n\) computed by fitting a two-component Gaussian mixture with unrestricted covariances to \(X_S\), minus the Bayesian information criterion penalty and plus the logarithm of the complexity prior with \(a=1.1\) and \(c=1\). The full-data criterion subtracts from this the penalized maximized log-likelihood \(\ell^G_n(S)\) of a single Gaussian on the same coordinates, the large-sample form of \eqref{eq:full-posterior}. Each search starts either from the empty set or from \(S_0\) and repeatedly makes the single addition or deletion that most increases the criterion, stopping when no move helps. Each setting is replicated \(200\) times. Table~\ref{tab:common-space} reports the size of the selected set \(\widehat S\), the proportion of \(S_0\) selected, the false-discovery proportion with near copies counted as false, the number of near copies selected, the proportion of the five relevant coordinates represented by themselves or by their near copy, and the adjusted Rand index of the mixture fitted to \(X_{\widehat S}\), set to zero when \(\widehat S\) is empty.

\begin{table}[!htbp]
\centering
\caption{Selection by the criterion of Section~\ref{sec:bayesian-distillation} as written and by the full-data criterion, with \(n=300\), \(K=2\), and \(s_0=5\). Averages over \(200\) replications; standard errors are at most 0.04 for proportions and 0.16 for set sizes. The oracle fits the mixture to \(X_{S_0}\).}
\label{tab:common-space}
\begin{tabular}{llcccccc}
\hline
Criterion & Start & \(|\widehat S|\) & TPR & FDP & Copies & Directions & ARI \\
\hline
\multicolumn{8}{l}{\(p=50\)} \\
As written & Empty & 0.00 & 0.00 & 0.00 & 0.00 & 0.00 & 0.00 \\
As written & \(S_0\) & 0.00 & 0.00 & 0.00 & 0.00 & 0.00 & 0.00 \\
Full data & Empty & 3.58 & 0.52 & 0.21 & 0.96 & 0.71 & 0.78 \\
Full data & \(S_0\) & 4.54 & 0.90 & 0.01 & 0.03 & 0.91 & 0.98 \\
Oracle & -- & 5.00 & 1.00 & 0.00 & 0.00 & 1.00 & 0.99 \\
\hline
\multicolumn{8}{l}{\(p=200\)} \\
As written & Empty & 0.00 & 0.00 & 0.00 & 0.00 & 0.00 & 0.00 \\
As written & \(S_0\) & 0.00 & 0.00 & 0.00 & 0.00 & 0.00 & 0.00 \\
Full data & Empty & 2.39 & 0.33 & 0.17 & 0.72 & 0.48 & 0.54 \\
Full data & \(S_0\) & 4.30 & 0.86 & 0.00 & 0.01 & 0.86 & 0.98 \\
Oracle & -- & 5.00 & 1.00 & 0.00 & 0.00 & 1.00 & 0.99 \\
\hline
\end{tabular}
\end{table}

The criterion as written selects no coordinate in any replication, from either starting point and for both values of \(p\), so no clustering is possible. Even when the search starts at \(S_0\), deleting relevant coordinates one at a time raises the criterion until the set is empty. This is the behavior predicted by Remark~\ref{rem:entropy}: the separation along each relevant coordinate is \(2.5\) within-group standard deviations, far below the thresholds in Remark~\ref{rem:entropy}.

The full-data criterion started at \(S_0\) keeps \(4.3\) to \(4.5\) of the five relevant coordinates on average, selects essentially no irrelevant coordinate, and attains an adjusted Rand index of \(0.98\), against \(0.99\) for the oracle. The occasional loss of a relevant coordinate occurs when its additional contribution, given the other four, is smaller than the penalty for one more coordinate. Started from the empty set, the same criterion does worse: the average adjusted Rand index is \(0.78\) for \(p=50\) and \(0.54\) for \(p=200\). A single coordinate with this separation is only weakly bimodal, so the first addition often fails to pay for its penalty, which grows with \(p\) through the complexity prior, and the greedy search then stops early. When the search does start, it sometimes picks a near copy in place of the original coordinate, which the Directions column credits and which costs little for clustering. This weakness belongs to the one-at-a-time search rather than to the criterion, and a posterior sampler that proposes moves of two coordinates, or a search started from a screened set, would reduce it. The comparison isolates the point of this section: the as-written criterion fails because of the dimension mismatch, whereas the full-data criterion recovers the subgroup structure when the search is given a reasonable start.

\section{Conclusion}
\label{sec:conclusion}

High-dimensional clustering is most useful when the recovered groups correspond to interpretable and reproducible sources of heterogeneity. In mixture models, this objective is difficult to achieve when the relevant subgroup information is sparse and embedded in a large number of nuisance variables. The framework developed here addresses this difficulty by placing variable selection before clustering, while preserving a coherent probabilistic link between the selected variables, the latent subgroup labels, and the response.

The proposed procedure has three features that are especially important in scientific applications. First, the distilled subspace is selected using posterior evidence, rather than by marginal association or purely geometric dimension reduction. Second, clustering is performed only after the high-dimensional noise has been attenuated, which improves both interpretability and statistical stability. Third, the subgroup-specific conditional-independence analysis provides a second layer of inference: after the subgroups are estimated, one may examine how the selected variables relate to one another within each subgroup.

The theory formalizes these points. Under the irrelevance condition, the relevant subspace preserves all subgroup information; under separation conditions and accurate estimation of the mixture parameters on that subspace, the clustering rule consistently recovers the latent labels. Under additional sparsity assumptions on subgroup-specific precision matrices, conditional-independence tests within the estimated subgroups are asymptotically valid. These results support the use of Bayesian distillation as a principled tool for high-dimensional subgroup discovery, particularly in settings such as genomics and multi-omics studies where interpretability, uncertainty quantification, and structural insight are all central to the analysis.

\bibliographystyle{plainnat}
\bibliography{references}
\newpage
\appendix

\section{Appendix A - Proofs of Results}
\label{app:proofs}

\subsection{Proof of Theorem~\ref{thm:information}}
\label{app:proof-information}

\begin{proof}[Proof of Theorem~\ref{thm:information}]
By Bayes' rule,
\[
\Prob(Z=k\mid X)
=
\frac{
\pi_k p(X\mid Z=k)
}{
\sum_{\ell=1}^K \pi_\ell p(X\mid Z=\ell)
}.
\]
Write \(X=(X_{S_0},X_{S_0^c})\). Then
\[
p(X\mid Z=k)
=
p(X_{S_0}\mid Z=k)
p(X_{S_0^c}\mid X_{S_0},Z=k).
\]
By Assumption~\ref{ass:irrelevance},
\[
p(X_{S_0^c}\mid X_{S_0},Z=k)
=
p(X_{S_0^c}\mid X_{S_0})
\]
for all \(k\). Therefore,
\[
\Prob(Z=k\mid X)
=
\frac{
\pi_k p(X_{S_0}\mid Z=k)p(X_{S_0^c}\mid X_{S_0})
}{
\sum_{\ell=1}^K
\pi_\ell p(X_{S_0}\mid Z=\ell)p(X_{S_0^c}\mid X_{S_0})
}.
\]
The common factor \(p(X_{S_0^c}\mid X_{S_0})\) cancels, so
\[
\Prob(Z=k\mid X)
=
\frac{
\pi_k p(X_{S_0}\mid Z=k)
}{
\sum_{\ell=1}^K
\pi_\ell p(X_{S_0}\mid Z=\ell)
}
=
\Prob(Z=k\mid X_{S_0}).
\]
Thus the selected subspace preserves all subgroup information.
\end{proof}

\subsection{Proof of Theorem~\ref{thm:clustering}}
\label{app:proof-clustering}

\begin{proof}[Proof of Theorem~\ref{thm:clustering}]
Write \(\xi_i=\Sigma_0^{-1/2}(X_{iS_0}-\mu_{Z_i0})\). Under the model the vectors \(\xi_1,\ldots,\xi_n\) are independent \(N_{s_0}(0,I)\) and independent of the labels. Let \(A_n\) be the intersection of the event in the accuracy condition, with its relabeling, and the event
\[
B_n=\Bigl\{\max_{1\le i\le n}\|\xi_i\|_2^2\le 2s_0+6\log n\Bigr\}.
\]
The chi-square tail bound \(\Prob(\|\xi_i\|_2^2\ge s_0+2\sqrt{s_0x}+2x)\le e^{-x}\) \citep{laurent2000adaptive}, applied with \(x=2\log n\) and \(2\sqrt{s_0x}\le s_0+x\), gives \(\Prob(B_n^c)\le n\cdot n^{-2}\to0\), so \(\Prob(A_n)\to1\).

On \(A_n\), \(\widehat S=S_0\) and \(\widehat Z_i=\argmax_k\widehat g_k(X_{iS_0})\), where
\[
\widehat g_k(x)=\log\widehat\pi_k-\tfrac12\log\det\widehat\Sigma_k-\tfrac12(x-\widehat\mu_k)^\top\widehat\Sigma_k^{-1}(x-\widehat\mu_k),
\]
and \(g_k\) denotes the same function evaluated at the true parameters. Fix \(i\) with \(Z_i=k\) and a component \(\ell\ne k\). Let \(D_{k\ell}=\|\Sigma_0^{-1/2}(\mu_{k0}-\mu_{\ell0})\|_2\ge\Delta_n\) and \(u_{i,k\ell}=D_{k\ell}^{-1}(\mu_{k0}-\mu_{\ell0})^\top\Sigma_0^{-1/2}\xi_i\), which is standard normal. Expanding the quadratic forms with \(X_{iS_0}=\mu_{k0}+\Sigma_0^{1/2}\xi_i\) gives
\[
g_k(X_{iS_0})-g_\ell(X_{iS_0})
=
\log\frac{\pi_{k0}}{\pi_{\ell0}}
+\tfrac12D_{k\ell}^2
+D_{k\ell}u_{i,k\ell}.
\]

Next, bound the effect of the estimation error. By Assumption~\ref{ass:eigenvalue}, once \(\varepsilon_n\le\min(c/2,1,\pi_{\min}/2)\), the eigenvalues of every \(\widehat\Sigma_m\) lie in \([c/2,2C]\) on \(A_n\), so \(\|\widehat\Sigma_m^{-1}-\Sigma_0^{-1}\|_{\mathrm{op}}\le2\varepsilon_n/c^2\) and \(|\log\det\widehat\Sigma_m-\log\det\Sigma_0|\le2s_0\varepsilon_n/c\). Since \(\pi_{m0}\ge\pi_{\min}\), \(|\log\widehat\pi_m-\log\pi_{m0}|\le C'\varepsilon_n\). Writing \(v=x-\mu_{m0}\) and \(\delta=\widehat\mu_m-\mu_{m0}\), the difference of the quadratic forms equals \(v^\top(\widehat\Sigma_m^{-1}-\Sigma_0^{-1})v-2\delta^\top\widehat\Sigma_m^{-1}v+\delta^\top\widehat\Sigma_m^{-1}\delta\). Combining these bounds, there is a constant \(C_0\), depending only on \(c\), \(C\) and \(\min_m\pi_{m0}\), such that on \(A_n\)
\[
|\widehat g_m(x)-g_m(x)|\le C_0\varepsilon_n\bigl(1+s_0+\|x-\mu_{m0}\|_2^2\bigr)
\qquad\text{for every }x\text{ and }m.
\]
Moreover \(\|X_{iS_0}-\mu_{k0}\|_2^2\le C\|\xi_i\|_2^2\) and \(\|X_{iS_0}-\mu_{\ell0}\|_2^2\le2CD_{k\ell}^2+2C\|\xi_i\|_2^2\). On \(A_n\), the total perturbation of \(g_k-g_\ell\) at \(X_{iS_0}\) is therefore at most \(C_3\varepsilon_n(1+3s_0+6\log n+D_{k\ell}^2)\) for a constant \(C_3\).

If \(\widehat Z_i\ne k\), then \(\widehat g_\ell(X_{iS_0})\ge\widehat g_k(X_{iS_0})\) for some \(\ell\ne k\). With \(L=\max_{k\ne\ell}|\log(\pi_{k0}/\pi_{\ell0})|\), this forces
\[
u_{i,k\ell}
\le
-D_{k\ell}\left\{\tfrac12-C_3\varepsilon_n-\frac{L+C_3\varepsilon_n(1+3s_0+6\log n)}{D_{k\ell}^2}\right\}
\le
-\Delta_n\left(\tfrac12-\eta_n\right),
\]
where \(\eta_n=C_3\varepsilon_n+\{L+C_3\varepsilon_n(1+3s_0+6\log n)\}/\Delta_n^2\). The second inequality uses \(D_{k\ell}\ge\Delta_n\) and holds once \(\eta_n<1/2\). Since \(\varepsilon_n\to0\), \(\Delta_n\to\infty\) and \(\varepsilon_n(s_0+\log n)=o(\Delta_n^2)\), the sequence \(\eta_n\) tends to zero. Hence, on \(A_n\),
\[
\ind(\widehat Z_i\ne Z_i)\le V_i,
\qquad
V_i=\sum_{\ell\ne Z_i}\ind\bigl\{u_{i,Z_i\ell}\le-\Delta_n(\tfrac12-\eta_n)\bigr\},
\]
for every \(i\), with labels understood after the relabeling in \(A_n\). Each \(V_i\) is a function of \((Z_i,\xi_i)\) alone, so it does not depend on the fitted parameters, and the Gaussian tail bound gives
\[
\E V_i\le(K-1)\exp\bigl\{-\tfrac12\Delta_n^2(\tfrac12-\eta_n)^2\bigr\}.
\]

For \(c_\ast<1/8\), we have \(\tfrac12(\tfrac12-\eta_n)^2>c_\ast\) for large \(n\). Markov's inequality then shows that \(n^{-1}\sum_iV_i=O_P\{\exp(-c_\ast\Delta_n^2)\}\). Since \(\Prob(A_n)\to1\), the first claim follows. For the second, the union bound gives
\[
\Prob(\exists i:\widehat Z_i\ne Z_i)
\le
\Prob(A_n^c)+n(K-1)\exp\bigl\{-\tfrac12\Delta_n^2(\tfrac12-\eta_n)^2\bigr\}.
\]
If \(\Delta_n^2\ge C_2\log n\) with \(C_2>8\), the last term is at most \((K-1)n^{1-C_2(1/2-\eta_n)^2/2}\), whose exponent tends to \(1-C_2/8<0\). This proves exact recovery up to label permutation.
\end{proof}

\subsection{Proof of Theorem~\ref{thm:contraction}}
\label{app:proof-contraction}

\begin{proof}[Proof of Theorem~\ref{thm:contraction}]
By assumption, the posterior places probability tending to one on the true model \(S_0\). It is therefore sufficient to prove the contraction rate conditional on \(S_0\).

Within the true model, the effective parameter dimension is of order \(Ks_0\). The log likelihood is locally quadratic around the true parameter, and the Fisher information matrix has eigenvalues bounded away from zero and infinity under the regularity assumptions. Hence, for parameters at Euclidean distance \(r_n\) from the truth,
\[
r_n=M\sqrt{\frac{s_0\log p}{n}},
\]
the likelihood ratio is bounded by
\[
\exp(-c n r_n^2)
=
\exp(-c M^2 s_0\log p)
\]
outside the ball of radius \(r_n\).

The prior assigns at least
\[
\exp(-C s_0\log p)
\]
mass to a neighborhood of the true parameter. Choosing \(M\) sufficiently large makes the likelihood decay outside the ball dominate the prior mass denominator. Standard testing arguments for parametric models with growing dimension then give
\[
\Pi(\|\beta-\beta_0\|_2>M r_n\mid D_n,S_0)\to0.
\]
Combining this with \(\Pi(S_0\mid D_n)\to1\) proves the result.
\end{proof}

\subsection{Proof of Theorem~\ref{thm:ci-tests}}
\label{app:proof-ci-tests}

\begin{proof}[Proof of Theorem~\ref{thm:ci-tests}]
By Theorem~\ref{thm:clustering}, with probability tending to one the estimated subgroup labels equal the true subgroup labels up to permutation. Conditional on this event, the sample used for subgroup \(k\) is an independent sample from
\[
N(\mu_{k0},\Sigma_{k0}).
\]

Under the sparsity condition
\[
d_k^2\frac{\log s_0}{n_k}\to0,
\]
the graphical-lasso or nodewise-regression estimator has elementwise error small enough for debiasing. The debiased estimator admits the expansion
\[
\sqrt{n_k}
\left(
\widehat\Omega_{k,ab}^{d}
-
\Omega_{k0,ab}
\right)
=
\frac1{\sqrt{n_k}}
\sum_{i:Z_i=k}
\xi_{i,ab}
+
o_P(1),
\]
where \(\xi_{i,ab}\) are mean-zero random variables with variance \(\sigma_{k,ab}^2\). The central limit theorem gives
\[
\frac{
\sqrt{n_k}
(
\widehat\Omega_{k,ab}^{d}
-
\Omega_{k0,ab}
)
}{
\widehat\sigma_{k,ab}
}
\rightsquigarrow N(0,1).
\]
Under the null \(\Omega_{k0,ab}=0\), this is exactly the stated result for \(T_{k,ab}\). The Benjamini-Hochberg false-discovery control follows because the null \(p\)-values are asymptotically uniform. The additional error caused by estimated labels is \(o(1)\), since exact recovery holds with probability tending to one.

For a nonzero edge, the statistic has mean
\[
\frac{
\sqrt{n_k}\Omega_{k0,ab}
}{
\sigma_{k,ab}
},
\]
which diverges by assumption. Therefore the rejection probability tends to one.
\end{proof}

\subsection{Proofs for Section~\ref{sec:common-space}}
\label{app:proof-common-space}

\begin{proof}[Proof of Lemma~\ref{lem:units}]
If \(j\notin S_\gamma\), column \(j\) does not enter \(L_n(\Theta,\gamma)\), and there is nothing to prove. Suppose \(j\in S_\gamma\) and let \(D\) be the diagonal matrix of order \(|S_\gamma|\) whose entry for coordinate \(j\) is \(c\) and whose other diagonal entries are one. For \(\Theta\in\mathcal T_\gamma\), let \(\Theta_c\) be obtained by replacing \(\mu_k\), \(\Sigma_k\) and \(\beta_k\) by \(D\mu_k\), \(D\Sigma_kD\) and \(D^{-1}\beta_k\) for every \(k\), leaving the other parameters unchanged. The map \(\Theta\mapsto\Theta_c\) is a bijection of \(\mathcal T_\gamma\). For every \(x\in\R^{|S_\gamma|}\),
\[
\varphi(Dx;D\mu_k,D\Sigma_kD)=|c|^{-1}\varphi(x;\mu_k,\Sigma_k),
\]
by the change-of-variables formula for densities, and \(r_k(y\mid Dx;\theta_{k,c})=r_k(y\mid x;\theta_k)\) because \((Dx)^\top D^{-1}\beta_k=x^\top\beta_k\). Each of the \(n\) factors of \(L_n\) therefore acquires the factor \(|c|^{-1}\), so that \(L_n^{(c)}(\Theta_c,\gamma)=|c|^{-n}L_n(\Theta,\gamma)\). Taking the supremum over \(\Theta\in\mathcal T_\gamma\) and logarithms proves the claim.
\end{proof}

\begin{proof}[Proof of Theorem~\ref{thm:full-data-posterior}]
Fix \(\gamma\) with \(S=S_\gamma\) and write \(T=S^c\). If \(S=\emptyset\), then \(g_\gamma\) is the \(N_p\) density, so \(m^{\mathrm{full}}_\gamma(D_n)=m^{\mathrm{mix}}_\gamma(D_n)\,m^G(X)\), which is \eqref{eq:full-marginal} with \(m^G(X_\emptyset)=1\). If \(T=\emptyset\), then \(g_\gamma\equiv1\) and \(m^G(X_S)=m^G(X)\), so \eqref{eq:full-marginal} again holds. Assume therefore that both are nonempty.

The first step is the a priori independence of \(\zeta_S=(m_S,V_{SS})\) and \(\vartheta_\gamma=(a_\gamma,B_\gamma,V_{T\mid S})\). Let \(W=V^{-1}\), which has a Wishart law under \eqref{eq:niw-prior}, and partition it by \((S,T)\). Then \(V_{SS}^{-1}=W_{SS}-W_{ST}W_{TT}^{-1}W_{TS}\), \(V_{T\mid S}=W_{TT}^{-1}\) and \(B_\gamma=-W_{TT}^{-1}W_{TS}\). For a Wishart matrix, the Schur complement \(W_{SS}-W_{ST}W_{TT}^{-1}W_{TS}\) is independent of \((W_{TS},W_{TT})\) \citep{geiger2002parameter}, so \(V_{SS}\) is independent of \((B_\gamma,V_{T\mid S})\). Conditionally on \(V\), the vector \(m\) is Gaussian with covariance \(V/\kappa_0\); hence \(m_S\sim N_s(m_{0,S},V_{SS}/\kappa_0)\), and \(a_\gamma=m_T-B_\gamma m_S\), the residual of the population regression of \(m_T\) on \(m_S\), is independent of \(m_S\) with law \(N(m_{0,T}-B_\gamma m_{0,S},V_{T\mid S}/\kappa_0)\). The conditional law of \(m_S\) depends on \(V\) only through \(V_{SS}\), and that of \(a_\gamma\) only through \((B_\gamma,V_{T\mid S})\). Combining these facts, \(\zeta_S\) and \(\vartheta_\gamma\) are independent. The marginal law of \(V_{SS}\) is \(\operatorname{IW}_s(\Psi_{0,SS},\nu_0-(p-s))\), a standard property of the inverse-Wishart law, which gives the stated prior for \(m^G(X_S)\).

Under \(N_p(m,V)\) the density of an observation factorizes as \(\varphi(x_S;m_S,V_{SS})\,g_\gamma(x_T\mid x_S;\vartheta_\gamma)\). By the independence just proved and Fubini's theorem,
\[
m^G(X)
=
\int\prod_{i=1}^n\varphi(X_{iS};\zeta_S)\,d\Pi(\zeta_S)
\int\prod_{i=1}^n g_\gamma(X_{iT}\mid X_{iS};\vartheta_\gamma)\,d\Pi(\vartheta_\gamma)
=
m^G(X_S)\,m^{\mathrm{reg}}_\gamma(X_T\mid X_S),
\]
where the second integral defines \(m^{\mathrm{reg}}_\gamma\). Under \eqref{eq:full-data-model}, \(\Theta\) and \(\vartheta_\gamma\) are a priori independent and the likelihood is \(L_n(\Theta,\gamma)\prod_ig_\gamma(X_{iT}\mid X_{iS};\vartheta_\gamma)\), so the same argument gives \(m^{\mathrm{full}}_\gamma(D_n)=m^{\mathrm{mix}}_\gamma(D_n)\,m^{\mathrm{reg}}_\gamma(X_T\mid X_S)\). Eliminating \(m^{\mathrm{reg}}_\gamma\) between the two identities proves \eqref{eq:full-marginal}.

The quantity \(m^G(X)\) is the marginal likelihood of \(X\) under the single Gaussian model with the proper prior \eqref{eq:niw-prior}, which does not involve \(\gamma\). It is finite for every data matrix, including when \(p>n\): integrating the normal--inverse-Wishart prior against the Gaussian likelihood gives the standard closed form, in which the posterior scale matrix is \(\Psi_0\) plus positive semidefinite terms and is therefore positive definite, so every factor is finite. Since \(m^G(X)\) is a common positive factor of all the marginal likelihoods, it cancels from Bayes' formula, which proves \eqref{eq:full-posterior}.
\end{proof}

\begin{proof}[Proof of Corollary~\ref{cor:invariance}]
The constraint is preserved by the transformation, since \(D\Sigma_k\Sigma_\ell^{-1}D^{-1}\) is similar to \(\Sigma_k\Sigma_\ell^{-1}\). The location shift \(b\) leaves both terms unchanged, since the mixture means and the Gaussian mean are unrestricted and the response model uses an unrestricted intercept. For the diagonal factor, apply the argument of Lemma~\ref{lem:units} one coordinate at a time: each coordinate \(j\in S_\gamma\) contributes \(-n\log|d_j|\) to \(\ell_n(\gamma)\), and the same reparametrization of the single Gaussian shows that it contributes \(-n\log|d_j|\) to \(\ell^G_n(S_\gamma)\). Coordinates outside \(S_\gamma\) contribute to neither term, so the difference is unchanged.
\end{proof}

\end{document}